\documentclass{amsart}
\usepackage{amssymb,amsmath,amscd,amsthm}
\usepackage{amsfonts,epsfig,latexsym,graphicx,amssymb}

\date{\today}

\newcommand{\Z}{{\mathbb Z}}
\newcommand{\R}{{\mathbb R}}
\newcommand{\C}{{\mathbb C}}

\newcommand{\T}{{\mathbb T}}

\newtheorem{theorem}{Theorem}[section]
\newtheorem{remark}[theorem]{Remark}
\newtheorem{lemma}[theorem]{Lemma}
\newtheorem{prop}[theorem]{Proposition}
\newtheorem{coro}[theorem]{Corollary}

\def\be{\begin{equation}}
\def\ee{\end{equation}}

\begin{document}

\title[Almost Ballistic Transport for the Fibonacci Hamiltonian]{Almost Ballistic Transport for the Weakly Coupled Fibonacci Hamiltonian}

\author[D.\ Damanik]{David Damanik}

\address{Department of Mathematics, Rice University, Houston, TX~77005, USA}

\email{damanik@rice.edu}

\thanks{D.\ D.\ was supported in part by a Simons Fellowship and NSF grant DMS--1067988.}

\author[A.\ Gorodetski]{Anton Gorodetski}

\address{Department of Mathematics, University of California, Irvine, CA~92697, USA}

\email{asgor@math.uci.edu}

\thanks{A.\ G.\ was supported in part by NSF grants DMS--1301515 and %DMS--0901627 and
IIS-1018433.}

\begin{abstract}
We prove estimates for the transport exponents associated with the weakly coupled Fibonacci Hamiltonian. It follows in particular that the upper transport exponent $\tilde \alpha^\pm_u$ approaches the value one as the coupling goes to zero. Moreover, for sufficiently small coupling, $\tilde \alpha^\pm_u$ strictly exceeds the fractal dimension of the spectrum.
\end{abstract}

\maketitle

\section{Introduction}\label{s.intro}

The Fibonacci Hamiltonian is the most prominent model in the study of electronic properties of quasicrystals. It is given by the discrete one-dimensional Schr\"odinger operator
\begin{equation}\label{e.fib}
[H_{\lambda,\omega} u](n) = u(n+1) + u(n-1) + \lambda \chi_{[1-\alpha,1)}(n \alpha + \omega \!\!\!\! \mod 1) u(n),
\end{equation}
acting in $\ell^2(\Z)$, where $\lambda > 0$ is the coupling constant, $\alpha = \frac{\sqrt{5}-1}{2}$ is the frequency, and $\omega \in \T = \R / \Z$ is the phase. In particular, $\alpha$ is the inverse of the golden ratio
$$
\phi = \frac{\sqrt{5} + 1}{2}.
$$

It is well known that the operator $H_{\lambda,\omega}$ has purely singular continuous spectrum for all parameter values; compare \cite{DL99a, S89}. The RAGE Theorem (see, e.g., \cite[Theorem~XI.115]{RS3}) therefore suggests that when studying the Schr\"odinger time evolution for this Schr\"odinger operator, that is, $e^{-itH_{\lambda,\omega}} \psi$ for some initial state $\psi \in \ell^2(\Z)$, one should consider time-averaged quantities. For simplicity, let us consider initial states of the form $\delta_n$, $n \in \Z$. Since a translation in space simply results in an adjustment of the phase, we may without loss of generality focus on the particular case $\psi = \delta_0$. The time-averaged spreading of $e^{-itH_{\lambda,\omega}} \delta_0$ is usually captured on a power-law scale as follows; compare, for example, \cite{DT10, L96}. For $p > 0$, consider the $p$-th moment of the position operator,
$$
\langle |X|_{\delta_0}^p \rangle (t) = \sum_{n \in \Z} |n|^p | \langle e^{-itH_{\lambda,\omega}} \delta_0 , \delta_n \rangle |^2
$$
We average in time as follows. If $f(t)$ is a function of $t > 0$ and $T > 0$ is given, we denote the time-averaged function at $T$ by $\langle f \rangle (T)$:
$$
\langle f \rangle (T) = \frac{2}{T} \int_0^{\infty} e^{-2t/T} f(t) \, dt.
$$
Then, the corresponding upper and lower transport exponents $\tilde \beta^+_{\delta_0}(p)$ and $\tilde \beta^-_{\delta_0}(p)$ are given, respectively, by
$$
\tilde \beta^+_{\delta_0}(p) = \limsup_{T \to \infty} \frac{\log \langle \langle |X|_{\delta_0}^p \rangle \rangle (T) }{p \, \log T},
$$
$$
\tilde \beta^-_{\delta_0}(p) = \liminf_{T \to \infty} \frac{\log \langle \langle |X|_{\delta_0}^p \rangle \rangle (T) }{p \, \log T}.
$$
The transport exponents $\tilde \beta^\pm_{\delta_0}(p)$ belong to $[0,1]$ and are non-decreasing in $p$ (see, e.g., \cite{DT10}), and hence the following limits exist:
\begin{align*}
\tilde \alpha_l^\pm & = \lim_{p \to 0} \tilde \beta^\pm_{\delta_0}(p), \\
\tilde \alpha_u^\pm & = \lim_{p \to \infty} \tilde \beta^\pm_{\delta_0}(p).
\end{align*}

Ballistic transport corresponds to transport exponents being equal to one, diffusive transport corresponds to the value $\frac12$, and vanishing transport exponents correspond to (some weak form of) dynamical localization. In all other cases, transport is called anomalous. The Fibonacci Hamiltonian has long been the primary candidate for an interesting model exhibiting anomalous transport, going back at least to the work of Abe and Hiramoto \cite{AH}. Many papers have been devoted to a study of the transport properties of the Fibonacci Hamiltonian; see, for example, \cite{BLS, D98, D05, DKL, DST, DT03, DT05, DT07, DT08, JL2, kkl}. For example, it is known that all the transport exponents defined above are strictly positive for all $\lambda > 0$, $\omega \in \T$; see \cite{DKL}. On the other hand, upper bounds for all the transport exponents were shown in \cite{DT07} for $\lambda > 8$ (see also \cite{BLS} for a somewhat weaker result). The exact large coupling asymptotics of $\tilde \alpha_u^\pm$ were identified in \cite{DT08}, where is was shown that
\begin{equation}\label{e.inftyapproach2}
\lim_{\lambda \to \infty} \tilde \alpha_u^\pm \cdot \log \lambda = 2 \log \phi,
\end{equation}
uniformly in $\omega \in \T$. In particular, the Fibonacci Hamiltonian indeed gives rise to anomalous transport for sufficiently large coupling. The behavior in the weak coupling regime, on the other hand, is poorly understood. It has long been expected that the transport exponents should be ``continuous at zero,'' that is, they should approximate the value one, which is the transport exponent associated with the free Schr\"odinger operator (i.e., the one obtained by setting $\lambda$ equal to zero). In particular, it has been expected that
\begin{equation}\label{e.zeroapproach}
\lim_{\lambda \to 0} \tilde \alpha_u^\pm = 1.
\end{equation}
Since $\tilde \alpha_u^\pm \le 1$ by general principles, one needs to establish lower bounds that approach one as $\lambda$ tends to zero. However, the previously known lower bounds for small coupling are far away from one, and hence they were quite obviously far from optimal; compare \cite{D98, DG11, DKL, DST, DT03, DT05, JL2} for these prior results. One of the primary reasons for this apparent gap between our understanding and the (conjecturally) correct result is that the dimension of the spectral measure of $\delta_0$ is poorly understood. Other spectral quantities have been shown to be ``continuous at zero.'' Namely, as $\lambda$ tends to zero, the dimension of the spectrum tends to one \cite{DG11}, and the dimension of the density of states measure tends to one as well \cite{DG12}. Tangentially, we note that also the optimal H\"older exponent of the integrated density of states approaches the value associated with the free case (namely one-half) as $\lambda$ tends to zero \cite{DG13}. The result for the density of states measure just mentioned gave renewed hope for a corresponding result for the transport exponents since this measure is the phase average of the spectral measure in question, and hence one could hope for phase-averaged transport to approach ballistic rates in the zero coupling limit. Alas, no such general connection is known and therefore this did not allow the authors of \cite{DG12} to conclude the desired result.

In this paper we present a new approach to obtaining lower bounds for the transport exponents associated with the weakly coupled Fibonacci Hamiltonian, which allows us to obtain the desired continuity result \eqref{e.zeroapproach}.

\begin{theorem}\label{t.main}
There is a constant $c > 0$ such that for $\lambda > 0$ sufficiently small, we have
$$
1 - c\lambda^2 \le \tilde \alpha_u^\pm \le 1,
$$
uniformly in $\omega \in \T$.
\end{theorem}

The proof of Theorem~\ref{t.main} also gives estimates for $\tilde \beta^\pm_{\delta_0}(p)$. Since the specific expressions we obtain are somewhat lengthy, we do note state them here and refer the interested reader to Section~\ref{s.qdynamics} (see Remark~\ref{r.betapest}).

\medskip

Recall that the spectrum of $H_{\lambda,\omega}$ is independent of $\omega$ and hence may be denoted by $\Sigma_\lambda$. It is known that $\Sigma_\lambda$ is a zero-measure Cantor set for every $\lambda > 0$; see \cite{S89}. In fact, its Hausdorff dimension is strictly between zero and one \cite{Can}. Moreover, the box counting dimension of $\Sigma_\lambda$ exists and is equal to the Hausdorff dimension for $\lambda$ sufficiently large \cite{DEGT} and for $\lambda$ sufficiently small \cite{DG09}. Last asked in \cite{L96} whether the upper box counting dimension of the spectrum bounds all transport exponents from above. The large coupling asymptotics for $\tilde \alpha_u^\pm$ in \eqref{e.zeroapproach}, established in \cite{DT08}, provided a negative answer to this question because \cite{DEGT} had obtained the following large coupling asymptotics for the dimension of the spectrum:
\begin{equation}\label{e.inftyapproach}
\lim_{\lambda \to \infty} \dim \Sigma_\lambda \cdot \log \lambda = \hat c \log \phi,
\end{equation}
where $\hat c$ is an explicit constant that is strictly less than two ($\hat c = \frac{\log(1+\sqrt{2})}{\log \phi} = 1.8...$). Indeed, \eqref{e.inftyapproach2} and \eqref{e.inftyapproach}  imply that for $\lambda$ sufficiently large, we must have $\tilde \alpha_u^{\pm} > \dim \Sigma_\lambda$.

Here we can show that the small coupling regime also provides examples for which transport may exceed the dimension of the spectrum.

\begin{coro}\label{c.lastquestion}
For sufficiently small $\lambda > 0$, we have
$$
\tilde \alpha_u^{\pm} > \dim \Sigma_\lambda,
$$
uniformly in $\omega \in \T$.
\end{coro}

\begin{proof}
It was shown in \cite{DG11} that there are constants $C_1,C_2 > 0$ such that
$$
1 - C_1 \lambda \le \dim \Sigma_\lambda \le 1 - C_2 \lambda
$$
for $\lambda > 0$ sufficiently small. Combining this with the quadratic lower bound for $\tilde \alpha_u^{\pm}$ provided by Theorem~\ref{t.main}, the result follows.
\end{proof}

\section{The Dynamical Formalism and the Key Lemma}\label{s.tracemap}

The central tool in the study of the Fibonacci Hamiltonian is the trace map. This was realized early on in the very first papers studying this model; see, for example, \cite{Cas, KKT, S87}. The numerous recent advances (e.g., \cite{Can, DG11, DG12, DG13}) have been made possible by the use of more sophisticated tools from dynamical systems theory to exploit this connection between spectral properties of the Fibonacci Hamiltonian and the dynamics of the trace map. As we mentioned in the introduction, the spectral issues behind the phenomenon of almost ballistic transport for the weakly coupled Fibonacci Hamiltonian are far from being understood precisely, and this has prevented the previous works on this problem from establishing lower bounds close to one for small $\lambda$. Here we propose a way around this issue. We will not seek to establish precise estimates for the dimensional properties of spectral measures, but rather rely on the approach developed by Damanik and Tcheremchantsev (e.g., \cite{DT03, DT07, DT08}) that is based the Parseval formula, and hence on integrals with respect to Lebesgue measure, rather than integrals with respect to spectral measures. Especially the later papers \cite{DT07, DT08} emphasized complex analysis methods, and led to a study of the complex trace map. Distortion results then play a crucial role, and the necessary input is given by certain estimates of the derivatives of the entries of the trace map with respect to the energy. Since the energies where the derivatives are considered are real, the necessary input may be established by focusing on the real trace map. Consequently, in this section we recall the standard setting in which the real trace map is studied and prove the key lemma about the size of the derivatives we need later. This will provide the input to the discussion in the complex setting which will be presented in the next section.

\subsection{The Trace Map}

There is a fundamental connection between the spectral properties of the Fibonacci Hamiltonian and the dynamics of the \textit{trace map}
$$
T : \Bbb{R}^3 \to \Bbb{R}^3, \; T(x,y,z)=(2xy-z,x,y).
$$
The function $G(x,y,z) = x^2+y^2+z^2-2xyz-1$, sometimes called the ``Fricke-Vogt invariant,'' is invariant under the action of $T$, and hence $T$ preserves the family of cubic surfaces
$$
S_\lambda = \left\{(x,y,z)\in \Bbb{R}^3 : x^2+y^2+z^2-2xyz=1+ \frac{\lambda^2}{4} \right\}.
$$
It is therefore natural to consider the restriction $T_{\lambda}$ of the trace map $T$ to the invariant surface $S_\lambda$. That is, $T_{\lambda}:S_\lambda \to S_\lambda$, $T_{\lambda}=T|_{S_\lambda}$. We denote by $\Omega_{\lambda}$ the set of points in $S_\lambda$ whose full orbits under $T_{\lambda}$ are bounded (it is known that $\Omega_\lambda$ is equal to the non-wandering set of $T_\lambda$).

\subsection{Hyperbolicity of the Trace Map}

Recall that an invariant closed set $\Lambda$ of a diffeomorphism $f : M \to M$ is \textit{hyperbolic} if there exists a splitting of the tangent space $T_xM=E^u_x\oplus E^u_x$ at every point $x\in \Lambda$ such that this splitting is invariant under $Df$, the differential $Df$ exponentially contracts vectors from the stable subspaces $\{E^s_x\}$, and the differential of the inverse, $Df^{-1}$, exponentially contracts vectors from the unstable subspaces $\{E^u_x\}$. A hyperbolic set $\Lambda$ of a diffeomorphism $f : M \to M$ is \textit{locally maximal} if there
exists a neighborhood $U$ of $\Lambda$ such that
$$
\Lambda=\bigcap_{n\in\Bbb{Z}}f^n(U).
$$
It is known that for $\lambda > 0$, $\Omega_{\lambda}$ is a locally maximal hyperbolic set of $T_{\lambda} : S_\lambda \to S_\lambda$; see \cite{Can, Cas, DG09}.

\subsection{Properties of the Trace Map for $\lambda=0$}\label{ss.vequzero}

The surface
$$
\mathbb{S} = S_0 \cap \{ (x,y,z)\in \Bbb{R}^3 : |x|\le 1, |y|\le 1, |z|\le 1\}
$$
is homeomorphic to $S^2$, invariant under $T$, smooth everywhere except at the four points $P_1=(1,1,1)$, $P_2=(-1,-1,1)$, $P_3=(1,-1,-1)$, and $P_4=(-1,1,-1)$, where $\mathbb{S}$ has conic singularities, and the trace map $T$ restricted to $\mathbb{S}$ is a factor of the hyperbolic automorphism of $\T^2 = \R^2 / \Z^2$ given by
\begin{equation}\label{e.semiconj1}
\mathcal{A}(\theta, \varphi) = (\theta + \varphi, \theta)\ (\text{\rm mod}\ 1).
\end{equation}
The semi-conjugacy is given by the map
\begin{equation}\label{e.semiconj2}
F: (\theta, \varphi) \mapsto (\cos 2\pi(\theta + \varphi), \cos 2\pi \theta, \cos 2\pi \varphi).
\end{equation}
The map $\mathcal{A}$ is hyperbolic, and is given by the matrix $A = \begin{pmatrix} 1 & 1 \\ 1 & 0 \end{pmatrix}$, which has eigenvalues $\phi$ and $- \phi^{-1}$.

\subsection{Spectrum and Trace Map}

Denote by $\ell_\lambda$ the line
$$
\ell_\lambda = \left\{ \left(\frac{E-\lambda}{2}, \frac{E}{2}, 1 \right) : E \in \Bbb{R} \right\}.
$$
It is easy to check that $\ell_\lambda \subset S_\lambda$. An energy $E \in \Bbb{R}$ belongs to the spectrum $\Sigma_\lambda$ of the Fibonacci Hamiltonian if and only if the positive semiorbit of the point $(\frac{E-\lambda}{2}, \frac{E}{2}, 1)$ under iterates of the trace map $T$ is bounded; see S\"ut\H{o}~\cite{S87}. Moreover, the stable manifolds of points in $\Omega_\lambda$ intersect the line $\ell_\lambda$ transversally if $\lambda > 0$ is sufficiently small \cite{DG09} or if $\lambda \ge 16$ \cite{Cas}. It is an open problem whether this transversality condition holds for all $\lambda > 0$.

S\"ut\H{o}'s theorem considers the iterates of $\left(\frac{E-\lambda}{2}, \frac{E}{2}, 1 \right)$ under $T_\lambda$. We may write
$$
T_\lambda^k \left(\frac{E-\lambda}{2}, \frac{E}{2}, 1 \right) = \left( x_k(E), x_{k-1}(E), x_{k-2}(E) \right).
$$
It turns out that, for $k \ge 0$, $x_k(E)$ is equal to one-half the trace of the transfer matrix associated with the zero phase Fibonacci Hamiltonian from the origin to the site $F_k$, where $F_k$ denotes the $k$-th Fibonacci number; see \cite{S87}. From either description it follows that $x_k$ is a polynomial of degree $F_k$. It is known that all its zeros are real and simple. In the next subsection we will study the size of the derivative of $x_k$ at one of these zeros. As we will see later, this will be in direct relation to the bounds on the transport exponents we wish to prove.

\subsection{The Key Lemma}

In this subsection we will prove the following lemma, which provides bounds on the size of the derivative of $x_k$ at one of its zeros in the small coupling regime. These bounds are the key ingredient in our proof of Theorem~\ref{t.main}.

\begin{lemma}\label{l.keylemma}
There is $\lambda_1 > 0$ such that the following statements hold.

{\rm (a)} For every $\varepsilon > 0$ and every $\lambda \in (0,\lambda_1)$, there is  $k_0 \in \Z_+$ such that for every $k \ge k_0$, there exists $E_k \in \R$ such that $x_k(E_k) = 0$ and
\begin{equation}\label{e.keyestimate}
-\varepsilon + d(\lambda) < \frac1k \log |x_k'(E_k)| < \varepsilon + d(\lambda),
\end{equation}
where
\begin{equation}\label{e.glambdadef}
d(\lambda) = \frac{1}{6}\log\left(\left(\frac{\lambda^4}{2}+4\lambda^2+9\right)+(4+\lambda^2)\sqrt{\frac{\lambda^4}{2}+2\lambda^2+5}\right).
\end{equation}
We have
\begin{equation}\label{e.glambdaest}
\lambda^2 \lesssim 1 - \frac{\log \phi}{d(\lambda)} \lesssim \lambda^2.
\end{equation}

{\rm (b)} For every $\lambda \in (0,\lambda_1)$ and every $\delta \in (0,\frac{\lambda^2}{4})$, the following holds. If $b_k \subset \R$ is the connected component of $x_k^{-1}([-1-\delta,1+\delta])$ that contains $E_k$ {\rm (}from part~{\rm (a))}, then $E_k$ is the only zero of $x_k$ in $b_k$.
\end{lemma}

The following lemma will be used in the proof of Lemma~\ref{l.keylemma}.

\begin{lemma}\label{l.keylemma2}
Let $f : M^2 \to M^2$ be a $C^2$-diffeomorphism with a fixed hyperbolic saddle point $s \in M^2$. Suppose that $\ell : \R \to M^2$ is a line parametrized by the parameter $E \in \R$. Let $g:M^2\to \mathbb{R}$ be a smooth function such that the level curve $L=\{g=0\}$ is non-singular. Suppose also that the level curve $L \subset M^2$ is  such that the stable manifold $W^s(s)$ intersects $\ell(\R)$ transversally and the unstable manifold $W^u(s)$ intersects $L$ transversally. Then, for all sufficiently large $k \in \Z_+$, there exists $E_k \in \R$ such that $f^k(\ell(E_k)) \in L$ and
$$
\lim_{k \to \infty} \frac1k \log \left| \frac{d}{dE} g(f^k(\ell(E)) )\Big|_{E_k} \right| = \log \mu^u,
$$
where $\mu^u$ is the unstable multiplier of $s$.
\end{lemma}

\begin{proof}
Set $r = W^s(s) \cap \ell(\R)$ and $q = W^u(s) \cap L$. Fix an arc $J$ of $W^u(s)$ that contains both $s$ and $q$. Due to the Inclination Lemma (see, e.g., \cite[Lemma~1.1]{P69}), for all sufficiently large $k \in \Z_+$, there is a small interval $I_k \subset \ell(\R)$, such that $r \in I_k$ and $f^k(I_k) \to J$ in the $C^1$ topology. This implies that for each large $k$, there exists a point $p_k \in f^k(I_k) \cap L$. Denote $p_k' = f^{-k}(p_k) \in I_k \subset \ell$, and set $E_k = \ell^{-1}(p_k')$.

Let us estimate $|\frac{d}{dE} g(f^k(\ell(E)))|_{E_k}|$. Consider a neighborhood $V$ of $s$ such that a linearizing coordinate exists in $V$, that is, there is a diffeomorphism $H : V \to U \subset \R^2 = \R_x \times \R_y$ such that $H(s) = 0$, $H(W^s_\mathrm{loc}(s) \cap V) = \R_y \cap U$, $H(W^u_\mathrm{loc} \cap V) = \R_x \cap U$, and $H \circ f = A \circ H$ with
$$
A = \begin{pmatrix} \mu^u & 0 \\ 0 & \mu^s \end{pmatrix},
$$
where $\mu^u$, $\mu^s$ are the unstable and stable multipliers of $Df$ at $s$, respectively. Choose $k_1,k_2 \in \Z_+$ such that $f^{k_1}(r) \in V$ and $f^{-k_2}(q) \in V$. Set
$$
K = \max (\|H\|_{C^1}, \|H^{-1}\|_{C^1}), \ F = \max (\|f\|_{C^1}, \|f^{-1}\|_{C^1}).
$$
Let us denote by $\gamma$ the angle between $L$ and $W^u(s)$ at $q$. Notice that $\gamma \ne 0$ (since $L\pitchfork W^u(s)$ at $q$). Then for large enough $k \in \mathbb{N}$, we have
$$
\left| \frac{d}{dE} g(f^k(\ell(E))) \Big|_{E_k} \right| \ge F^{-k_1} \cdot K^{-1} \cdot \left( \mu^u \right)^{k - k_1 - k_2} \cdot K^{-1} \cdot F^{-k_2}\cdot \frac{1}{2}\left|\text{\rm grad\ }g(q)\right|\cdot |\sin \gamma|
$$
and
$$
\left| \frac{d}{dE} g(f^k(\ell(E))) \Big|_{E_k} \right| \le F^{k_1} \cdot K \cdot \left( \mu^u \right)^{k - k_1 - k_2} \cdot K \cdot F^{k_2} \cdot 2 \left|\text{\rm grad\ }g(q)\right|, %\left| \frac{d}{dE} g(f^n(\ell(E))) \Big|_{E_n} \right| \le \|f\|_{C^1}^{n_1} \cdot \|H\|_{C^1} \cdot \left( \mu^u \right)^{n - n_1 - n_2} \cdot \|H^{-1}\|_{C^1} \cdot \|f\|_{C^1}^{n_2},
$$
and the result follows.
\end{proof}

\begin{figure}[htb]
\includegraphics[width=0.9\textwidth]{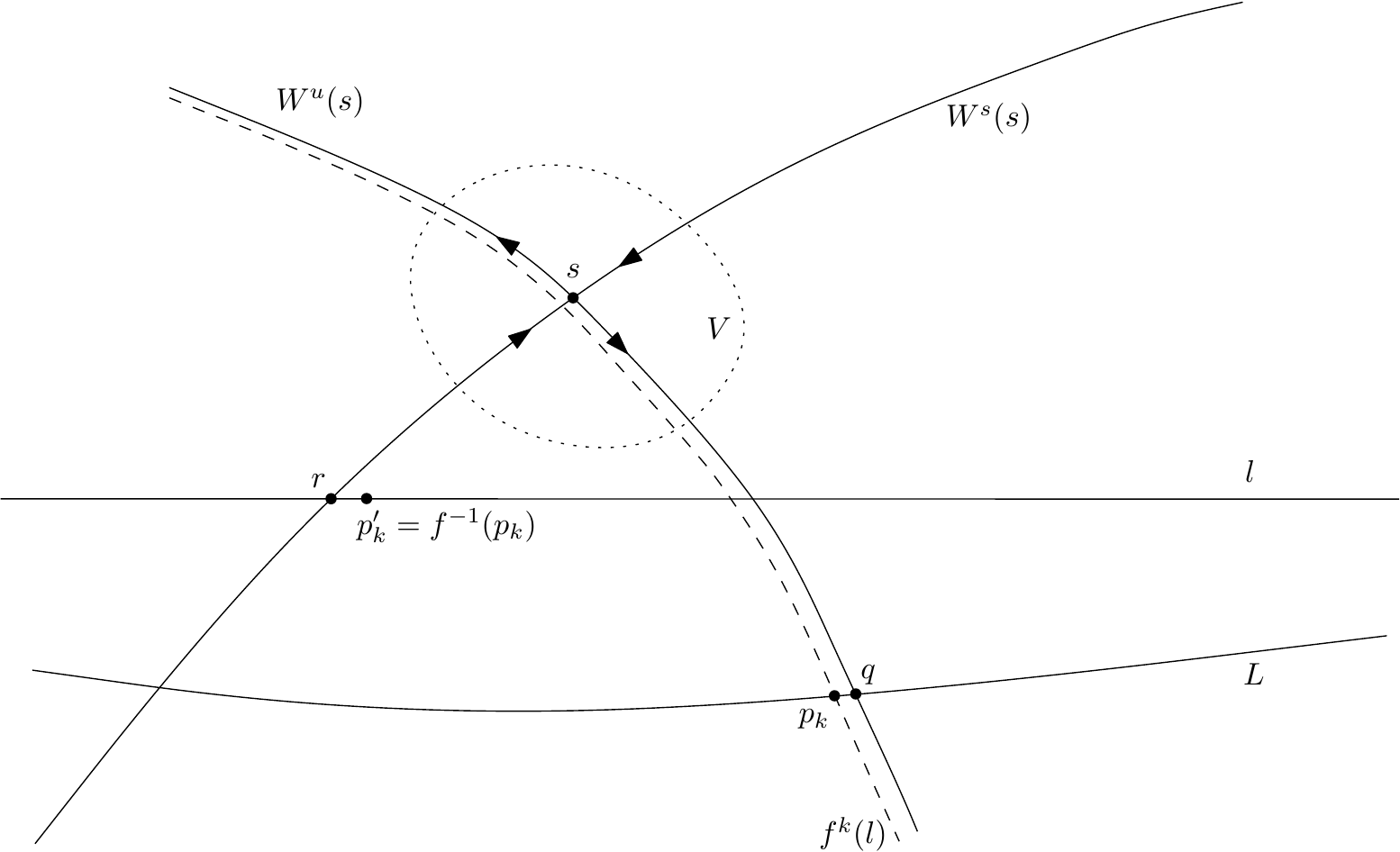}
\caption{Proof of Lemma \ref{l.keylemma2}}\label{fig.1}
\end{figure}

\begin{lemma}\label{l.mult}
Denote by $\mu^u(\lambda)$ the largest multiplier {\rm (}averaged over the orbit{\rm )} of the periodic orbit $\left(0, 0, \sqrt{1 + \frac{\lambda^2}{4}}\right)\in S_\lambda$ of the map $T_\lambda:S_\lambda\to S_\lambda$. Then
$$
\log \mu^u(\lambda)=\frac{1}{6}\log\left(\left(\frac{\lambda^4}{2}+4\lambda^2+9\right)+(4+\lambda^2)\sqrt{\frac{\lambda^4}{2}+2\lambda^2+5}\right)
$$
\end{lemma}
\begin{proof}%[Proof of Lemma~\ref{l.mult}.]
For $a \in \R$, note that $T^6(0,0,a) = (0,0,a)$. On $S_\lambda$, consider the corresponding six-cycle for $a^2 = 1 + \frac{\lambda^2}{4}$. We have
$$
DT^6(0,0,a) = \begin{pmatrix} 16 a^4 - 4 a^2 + 1 & 8 a^3 & 0 \\ 8 a^3 & 4 a^2 + 1 & 0 \\ 0 & 0 & 1 \end{pmatrix}.
$$

An explicit calculation shows that the matrix $DT^6(0,0,a)$ has eigenvalues
$$
\{1, 8a^4+1+ 4a^2\sqrt{4a^4+1}, 8a^4+1- 4a^2\sqrt{4a^4+1}\}.
$$
This proves Lemma \ref{l.mult}.
\end{proof}

\begin{remark}\label{r.zerovalue}
For $a = 1$ {\rm (}i.e., $\lambda = 0${\rm )}, we have
$$
DT^6(0,0,1) = \begin{pmatrix} 13 & 8 & 0 \\ 8  & 5 & 0 \\ 0 & 0 & 1 \end{pmatrix},
$$
which has eigenvalues $\{ 1 , 9 - 4 \sqrt{5} , 9 + 4 \sqrt{5} \}$.

In particular, $\log \mu^u(\lambda)\to \log \phi$ as $\lambda\to 0$.
\end{remark}

\begin{proof}[Proof of Lemma~\ref{l.keylemma}.]
(a) Let us apply Lemma~\ref{l.keylemma2} with $M^2 = S_\lambda$, $f = T_\lambda$, and
%\begin{align*}
%s & = \left( 0, 0, \sqrt{1 + \frac{\lambda^2}{4} }\right), \\
%\ell(E) & = \left( \frac{E - \lambda}{2}, \frac{E}{2} , 1 \right), \\
%g  (x, y, z) & =z, \\
%L = L(\lambda) & = S_\lambda \cap \{ (x, y, 0) : x,y \in \R \}.
%\end{align*}
$$
s = \left( 0, 0, \sqrt{1 + \frac{\lambda^2}{4} }\right),
$$
$$
\ell(E) = \left( \frac{E - \lambda}{2}, \frac{E}{2} , 1 \right),
$$
$$
g  (x, y, z) =z,
$$
$$
L = L(\lambda) = S_\lambda \cap \{ (x, y, 0) : x,y \in \R \}.
$$
For all sufficiently small $\lambda > 0$, the transversality conditions $W^s(s) \pitchfork \ell(\R)$ and $W^u(s) \pitchfork L$ hold. For these values of $\lambda$, the assertion of part (a) follows from Lemma~\ref{l.keylemma2} combined with Lemma \ref{l.mult} and Remark~\ref{r.zerovalue}.

(b) For the proof of this part, we need to use more specific properties of the trace map $T_\lambda$. Consider first the map $T_0 : \mathbb{S} \to \mathbb{S}$ and the periodic point $(0,0,1) \in \mathbb{S}$. The corresponding $6$-cycle arises from the $6$-cycle
$$
\Big( \frac14 , 0 \Big) \mapsto \Big( \frac14 , \frac14 \Big) \mapsto \Big( \frac12 , \frac14 \Big) \mapsto \Big( \frac34 , \frac12 \Big) \mapsto \Big( \frac14 , \frac34 \Big) \mapsto \Big( 0 , \frac14 \Big) \mapsto \Big( \frac14 , 0 \Big)
$$
of the map \eqref{e.semiconj1} under the semi-conjugacy \eqref{e.semiconj2}. Denote the points in the $T_0$-orbit of $(0,0,1)$ by $s_i(0)$, $i = 1, \ldots, 6$. They all belong to the cube $C_0 = \{ (x,y,z) : |x|, |y|, |z| \le 1 \}$. For each $s_i(0)$, choose a piece of $W^u(s_i(0))$ that contains $s_i(0)$, intersects $L(0) = \mathbb{S} \cap \{ (x, y, 0) : x,y \in \R \}$ exactly once, and connects the top of the cube $C_0$ with the bottom of the cube $C_0$. Such a piece of $W^u(s_i(0))$ will be called \emph{good}. Note that a good piece of $W^u(s_i(0))$ intersects $L(0)$ transversally.

Now consider the map $T_\lambda : S_\lambda \to S_\lambda$ for $\lambda > 0$. The point $(0,0,\frac{\lambda^2}{4})$ gives rise to a $6$-cycle consisting of the point $s_i(\lambda)$, $i = 1, \ldots, 6$. A finite piece of an invariant manifold will be changing continuously in $\lambda$, and hence we can consider a continuation of each of the good pieces of $W^u(s_i(0))$ chosen above. These pieces of $W^u(s_i(\lambda))$ will still be good for $\lambda > 0$ sufficiently small in the following sense: they intersect $L(\lambda) = S_\lambda \cap \{ (x, y, 0) : x,y \in \R \}$ exactly once (and this intersection is transversal) and they connect the top of the cube $C_\delta = \{ (x,y,z) : |x|, |y|, |z| \le 1+\delta \}$ with the bottom of the cube $C_\delta$. (Here we use that, by assumption, we consider $\delta \in (0,\frac{\lambda^2}{4})$ and that the intersection of the plane $\{ (x,y,z) : z = \frac{\lambda^2}{4} \}$ (resp., $\{ (x,y,z) : z = -\frac{\lambda^2}{4} \}$) and $S_\lambda$ consists of a pair of lines that are transversal to the continuations of the good pieces of unstable manifolds.)

Now apply Lemma~\ref{l.keylemma2} to the line $\ell_\lambda : \R \to S_\lambda$ and choose the interval $b_k \subset \R$ that corresponds to the intersection of $T_\lambda^k(\ell_\lambda)$ with $L(\lambda)$ and is close to one of the arcs of the unstable manifolds. The image $T_\lambda^k(\ell_\lambda(b_k))$ is exactly the part of $T_\lambda^n(\ell_\lambda(\R))$ that connects the top and the bottom of the cube $C_\delta$, and is $C^1$-close to the corresponding continuation of a good arc. This implies that $T_\lambda^k(\ell_\lambda(b_k))$ intersects $L(\lambda)$ at exactly one point and part~(b) follows.
\end{proof}

\section{The Damanik-Tcheremchantsev Setup and the Conclusion of the Proof of Theorem~\ref{t.main}}\label{s.qdynamics}

The purpose of this section is to show how Theorem~\ref{t.main} may be established with the help of the key dynamical lemma, Lemma~\ref{l.keylemma}. To this end, we first recall the general setup that was used in \cite{DT07, DT08} and then show how the estimate from Lemma~\ref{l.keylemma} enters the argument.

The Parseval identity implies (see, e.g., \cite[Lemma~3.2]{kkl})
\begin{equation}\label{e.parsform}
2\pi \int_0^{\infty} e^{-2t/T} | \langle \delta_n , e^{-itH} \delta_0 \rangle |^2 \, dt = \int_{-\infty}^\infty \left|\langle \delta_n  , (H - E - \tfrac{i}{T})^{-1} \delta_0 \rangle \right|^2 \, dE,
\end{equation}
and hence for the time averaged outside probabilities, defined by
\begin{equation}\label{e.taop}
\langle P(N,\cdot) \rangle (T) = \frac{2}{T} \int_0^{\infty} e^{-2t/T} \sum_{|n| \ge N} | \langle \delta_n , e^{-itH} \delta_0 \rangle |^2 \, dt,
\end{equation}
we have
\begin{equation}\label{e.taopresform}
\langle P(N,\cdot) \rangle (T) = \frac{1}{\pi T} \sum_{|n| \ge N} \int_{-\infty}^\infty \left|\langle \delta_n  , (H - E - \tfrac{i}{T})^{-1} \delta_0 \rangle \right|^2 \, dE.
\end{equation}
The right-hand side of \eqref{e.taopresform} may be studied by means of transfer matrices at complex energies, which are defined as follows. For $z \in \C$, $n \in \Z$, we set
$$
M(n;\omega,z) = \begin{cases} T(n;\omega,z) \cdots T(1;\omega,z) & n \ge 1, \\ T(n;\omega,z)^{-1} \cdots T(-1;\omega,z)^{-1} & n \le -1, \end{cases}
$$
where
$$
T(\ell;\omega,z) = \begin{pmatrix} z - \lambda \chi_{[1-\alpha,1)}(\ell \alpha + \omega \!\!\!\! \mod 1) & -1 \\ 1 & 0 \end{pmatrix}.
$$

For $\delta \ge 0$, consider the sets
$$
\sigma_k^\delta = \{ z \in \C : |x_k(z)| \le 1 + \delta \}.
$$

The following is \cite[Proposition~2]{DT08}:

\begin{prop}\label{p.tmest}
For every $\lambda , \delta > 0$, there are constants $C,\xi$ such that for every $k$, every $z \in \sigma_k^\delta$, and every $\omega \in \T$, we have
\begin{equation}\label{powerlaw}
\| M(n;\omega,z) \| \le C n^\xi.
\end{equation}
for $1 \le n \le F_k$.
\end{prop}

Combining ideas from the proof of \cite[Proposition~2]{DT08} and the proof of \cite[Theorem~5.1]{DG11}, one can show the following for the exponent $\xi$ in Proposition~\ref{p.tmest}. If we denote the largest root of the polynomial $x^3 - (2+\lambda) x - 1$ by $a_\lambda$ (note that for small $\lambda > 0$, we have $a_\lambda \approx \phi + c\lambda$ with a suitable constant $c$), then for any
\begin{equation}\label{e.gammaest}
\xi > 2 \frac{\log [(5 + 2\lambda)^{1/2} (3 + \lambda) a_\lambda]}{\log \phi},
\end{equation}
there is a constant $C$ such that \eqref{powerlaw} holds for $z \in \sigma_k^\delta$ and $\omega \in \T$.

\begin{proof}[Proof of Theorem~\ref{t.main}.]
Let us now consider $\lambda \in (0,\lambda_2)$ with $\lambda_2$ from Lemma~\ref{l.keylemma}. Fix $\delta \in (0,\lambda^2/8)$ and $\varepsilon > 0$. Lemma~\ref{l.keylemma} then implies that there is $k_0$ such that for every $k \ge k_0$, there is $E_k \in \R$ such that $x_k(E_k)= 0$ (so that in particular $E_k \in \sigma_k^{2\delta}$) and
\begin{equation}\label{e.xkprimeatek}
\left( e^{-\varepsilon + d(\lambda)} \right)^k \le |x_k'(e_k)| \le \left( e^{\varepsilon + d(\lambda)} \right)^k.
\end{equation}
Moreover, $E_k$ is the only zero of $x_k$ in its connected component relative to the set $\sigma_k^{2\delta} \cap \R$ (as a subset of $\R$).

Let us argue that $E_k$ is also the only zero of $x_k$ in its connected component relative to the set $\sigma_k^{2\delta}$ (as a subset of $\C$). Suppose this fails. Note that $\sigma_k^{2\delta}$ is symmetric with respect to the reflection about the real axis. If the connected component of $E_k$ relative to $\sigma_k^{2\delta} \cap \R$ extends to a connected component relative to $\sigma_k^{2\delta}$ that contains another zero of $x_k$, and hence another connected component of $\sigma_k^{2\delta} \cap \R$, we find that the boundary of this connected component, on which $x_k$ has constant modulus $1 + 2 \delta$, contains a closed curve that bounds a bounded region containing points at which $x_k$ has modulus strictly larger than $1 + 2 \delta$ (e.g., points on the real line strictly between the two connected components of $\sigma_k^{2\delta} \cap \R$ in question). Thus, we obtain a contradiction due to the maximum modulus principle.

Denote the connected component of $\sigma_k^{2\delta}$ that contains $E_k$ by $C_k$. Obviously, $E_k$ is also the only zero of $x_k$ in its connected component relative to $\sigma_k^{\delta}$, and we denote the connected component of $\sigma_k^{\delta}$ that contains $E_k$ by $D_k$. We can now proceed as in the proof of \cite[Proposition~3]{DT07} to show that $D_k$ must contain a ball centered at $E_k$ of a certain radius. For the convenience of the reader, we explain how this derivation works. Since $C_k$ contains exactly one zero of $x_k$, it follows from the maximum modulus principle and Rouch\'e's Theorem that
$$
x_k : \mathrm{int}(C_k) \to B(0,1 + 2\delta)
$$
is univalent, and hence
$$
x_k^{-1} : B(0,1 + 2\delta) \to \mathrm{int}(C_k)
$$
is well-defined and univalent as well. Consequently, the following mapping is a Schlicht function:
$$
F : B(0,1) \to \C, \quad F(z) = \frac{x_k^{-1} ((1 + 2\delta)z) - E_k}{(1 + 2\delta) [(x_k^{-1})'(0)]}.
$$
That is, $F$ is a univalent function on $B(0,1)$ with $F(0) = 0$ and $F'(0) = 1$.

The Koebe Distortion Theorem (see \cite[Theorem~7.9]{c}) implies that
\begin{equation}\label{koebe}
\frac{|z|}{(1 + |z|)^2} \le |F(z)| \le \frac{|z|}{(1 - |z|)^2} \text{ for } |z| \le 1.
\end{equation}
Evaluate the bound \eqref{koebe} on the circle $|z| = \frac{1 + \delta}{1 + 2\delta}$. For such $z$, we obtain
$$
\frac{(1 + \delta)(1 + 2\delta)}{(2 + 3 \delta)^2} \le |F(z)| \le \frac{(1 + \delta)(1 + 2\delta)}{\delta^2}.
$$
By definition of $F$ this means that
$$
| x_k^{-1} ((1 + 2\delta)z) - E_k | \le \frac{(1 + \delta)(1 + 2\delta)}{\delta^2} (1 + 2\delta) |(x_k^{-1})'(0)|
$$
and
$$
| x_k^{-1} ((1 + 2\delta)z) - E_k | \ge \frac{(1 + \delta)(1 + 2\delta)}{(2 + 3\delta)^2} (1 + 2\delta) |(x_k^{-1})'(0)|
$$
for all $z$ with $|z| = \frac{1 + \delta}{1 + 2\delta}$. In other words, if $|z| = 1 + \delta$, then
\begin{equation}\label{variation}
| x_k^{-1} (z) - E_k | \le \frac{(1 + \delta)(1 + 2\delta)^2}{\delta^2} |(x_k^{-1})'(0)|
\end{equation}
and
\begin{equation}\label{variation2}
| x_k^{-1} (z) - E_k | \ge \frac{(1 + \delta)(1 + 2\delta)^2}{(2 + 3\delta)^2} |(x_k^{-1})'(0)|.
\end{equation}
Since $|(x_k^{-1})'(0)| = |x_k'(E_k)|^{-1}$, we obtain from \eqref{e.xkprimeatek} and \eqref{variation} that
\begin{equation}\label{distance1}
| x_k^{-1} (z) - E_k | < \left(\frac{(1 + \delta)(1 + 2\delta)}{\delta} \right)^2 \left( e^{-\varepsilon + d(\lambda)} \right)^{-k}
\end{equation}
for all $z$ of magnitude $1 + \delta$. Similarly, \eqref{e.xkprimeatek} and \eqref{variation2} give
\begin{equation}\label{distance2}
| x_k^{-1} (z) - E_k | > \frac{(1 + \delta)(1 + 2 \delta)^2}{(2 + 3\delta)^2} \left( e^{\varepsilon + d(\lambda)} \phi \right)^{-k}
\end{equation}
for these values of $z$. Note that as $z$ runs through the circle of radius $1 + \delta$ around zero, the point $x_k^{-1} (z)$ runs through the entire boundary of $D_k$. Thus, \eqref{distance1} and \eqref{distance2} yield the following distortion result:
\begin{equation}\label{e.distortion}
B \Big( E_k, \frac{(1 + \delta)(1 + 2 \delta)^2}{(2 + 3\delta)^2} \left( e^{\varepsilon + d(\lambda)} \right)^{-k} \Big) \subseteq D_k \subseteq B \Big( E_k, \left(\frac{(1 + \delta)(1 + 2\delta)}{\delta} \right)^2 \left( e^{-\varepsilon + d(\lambda)} \right)^{-k} \Big).
\end{equation}
In particular, let us denote the radius of the inscribed ball by $r_k$:
$$
r_k = \frac{(1 + \delta)(1 + 2 \delta)^2}{(2 + 3\delta)^2} \left( e^{\varepsilon + d(\lambda)} \right)^{-k}.
$$
For $\rho > 0$ arbitrary, consider
\begin{equation}\label{e.sdef}
s = \frac{\lim_{k \to \infty} \frac1k \log \frac{1}{r_k}}{\log \phi} + \rho = \frac{\varepsilon + d(\lambda)}{\log \phi} + \rho.
\end{equation}
Then, for suitably chosen $C_\delta > 0$, we have $C_\delta F_k^{s} \ge \frac{2}{r_k}$ for every $k \ge 0$.

By Proposition~\ref{p.tmest} we have for $z \in D_k$ and $1 \le |n| \le F_k$ (use the uniformity in $\omega$ of the statement in Proposition~\ref{p.tmest} to deduce the analogous estimates on the left half-line),
\begin{equation}\label{e.2sidedpowerlaw}
\sup_{\omega \in \T} \| M(n;\omega,z) \| \le C |n|^\xi
\end{equation}
with suitable constants $C$ and $\xi$.

Take $N = F_k$ and consider $T \ge C_\delta N^{s}$ (which in turn implies $T \ge \frac{2}{r_k}$ by the choices of $s$ and $C_\delta$). Due to the Parseval formula \eqref{e.parsform}, we can bound the time-averaged outside probabilities from below as follows,
\begin{equation}\label{parseval}
\langle P(N,\cdot) \rangle (T) \gtrsim \frac1T \int_\R \left( \max \left\{ \|M(N;\omega,E+i/T)\|, \|M(-N;\omega,E+i/T)\| \right\} \right)^{-2} \, dE.
\end{equation}
See, for example, the proof of \cite[Theorem~1]{DT03} for an explicit derivation of \eqref{parseval} from \eqref{e.parsform}.

To bound the integral from below, we integrate only over those $E \in (E_k-r_k, E_k+r_k)$ for which $E+i/T \in B(E_k, r_k) \subset D_k$. Since $1/T \le r_k/2$, the length of such an interval $I_k$ is larger than $cr_k$ for some suitable $c>0$. For $E \in I_k$, we have
$$
\|M(N;\omega, E+i/T)\| \lesssim N^{\xi} \lesssim T^{\frac{\xi}{s}}.
$$
Therefore, \eqref{parseval} together with \eqref{e.2sidedpowerlaw} gives
\begin{equation}\label{eq1}
\langle P(N,\cdot) \rangle (T) \gtrsim \frac{r_k}{T} \, T^{-\frac{2\xi}{s}} \gtrsim T^{-2-\frac{2\xi}{2}},
\end{equation}
where $N = F_k$, $T \ge C_\delta N^s$, for any $k \ge k_0$.

Now let us take any sufficiently large $T$ and choose $k$ maximal with $C_\delta  F_k^s \le T$. Then,
$$
C_\delta F_k^s \le T < C_\delta F_{k+1}^s \le C_\delta 2^s F_k^s.
$$
It follows from \eqref{eq1} that
$$
\left\langle P \left( \tfrac{1}{2 C_\delta^{1/s}} T^{\frac{1}{s}},\cdot \right) \right\rangle (T) \ge \langle P(F_k,\cdot) \rangle (T) \gtrsim T^{-2-\frac{2\xi}{s}}
$$
for all sufficiently large $T$. It follows from the definition of $\langle \beta^-(p) \rangle$ and $\langle \alpha_u^- \rangle$ that
$$
\langle \beta_{\delta_0}^-(p) \rangle \ge \frac{1}{s} - \frac{2}{p}\left( 1 + \frac{\xi}{s} \right) = \left( \frac{\varepsilon + d(\lambda)}{\log \phi} + \rho \right)^{-1} - \frac{2}{p}\left( 1 + \xi \left( \frac{\varepsilon + d(\lambda)}{\log \phi} + \rho \right)^{-1} \right)
$$
and
$$
\langle \alpha_u^- \rangle \ge \frac{1}{s} = \left( \frac{\varepsilon + d(\lambda)}{\log \phi} + \rho \right)^{-1},
$$
by \eqref{e.sdef}. Since this is true for every $\varepsilon > 0$ and every $\rho > 0$, we have
\begin{equation}\label{e.betapest}
\langle \beta_{\delta_0}^-(p) \rangle \ge \frac{\log \phi}{d(\lambda)} - \frac{2}{p}\left( 1 + \frac{\xi \log \phi}{d(\lambda)} \right)
\end{equation}
and
\begin{equation}\label{e.alphauest}
\langle \alpha_u^- \rangle \ge \frac{1}{s} = \frac{\log \phi}{d(\lambda)}.
\end{equation}
Invoking \eqref{e.glambdaest}, the estimate \eqref{e.alphauest} yields $\langle \alpha_u^- \rangle \ge 1 - c \lambda^2$ for $\lambda$ sufficiently small.

Since we always have
$$
\langle \alpha_u^- \rangle (\lambda) \le \langle \alpha_u^+ \rangle (\lambda) \le 1,
$$
this completes the proof of the theorem.
\end{proof}

\begin{remark}\label{r.betapest}
Combining \eqref{e.betapest} with \eqref{e.glambdadef} and Proposition~\ref{p.tmest} along with \eqref{e.gammaest}, we obtain an explicit estimate for $\langle \beta_{\delta_0}^-(p) \rangle$.
\end{remark}

\end{document}